\documentclass[11pt,letterpaper]{article}
\usepackage[top=1in,left=1in,bottom=1in,right=1in]{geometry}

\usepackage{cite} % Make references as [1-4], not [1,2,3,4]
\usepackage{url}  % Formatting web addresses
\usepackage[colorlinks=true]{hyperref}
\usepackage{amsmath} 
\usepackage{amsthm} 
\usepackage{amssymb} 
\usepackage{algorithm}
\usepackage{algpseudocode}
\usepackage{xfrac,nicefrac}
\usepackage{graphicx}
\graphicspath{{./graphics/}}

\usepackage{tikz}
\usetikzlibrary{decorations.pathreplacing,calc,intersections,positioning,shapes.geometric}
\usepackage{forest}
\usepackage{subcaption}

\usepackage[disable,colorinlistoftodos,textsize=tiny]{todonotes}
\usepackage{xspace}

\newtheorem{theorem}{Theorem}%[section]
\newtheorem{observ}{Observation}%[section]

\newtheorem{corollary}{Corollary}
\newtheorem{lemma}{Lemma}

\newcommand{\lca}{\mathsf{lca}}

\DeclareMathOperator{\Le}{\mathsf{Le}\xspace}

\DeclareMathOperator{\cL}{\mathcal{L}\xspace}

\DeclareMathOperator{\cG}{\mathcal{G}\xspace}

\DeclareMathOperator{\cI}{\mathcal{I}}

\DeclareMathOperator{\qab}{\mathsf{ab|cd}}
\DeclareMathOperator{\qac}{\mathsf{ac|bd}}
\DeclareMathOperator{\qad}{\mathsf{ad|bc}}

\newcommand{\Q}{\mathbb{Q}}

\newcommand{\set}[1]{\{#1\}}

\begin{document}
%\mainmatter  

%\begin{frontmatter}

\title{Quartet-Based Inference Methods are Statistically Consistent Under the Unified Duplication-Loss-Coalescence Model}

\newcommand{\email}{}

\author{
Alexey Markin and Oliver Eulenstein\\
Department of Computer Science, \\
Iowa State University\\
Ames, IA, 50011, USA\\
\email{amarkin@iastate.edu $|$ oeulenst@iastate.edu}
}

\date{}

\maketitle

\begin{abstract}
	The classic multispecies coalescent (MSC) model provides the means for theoretical justification of incomplete lineage sorting-aware species tree inference methods. A large body of work in phylogenetics is dedicated to the design of inference methods that are statistically consistent under MSC. One of such particularly popular methods is ASTRAL, a quartet-based species tree inference method. A few recent studies suggested that ASTRAL also performs well when given multi-locus gene trees in simulation studies. Further, Legried et al. recently demonstrated that ASTRAL is statistically consistent under the gene duplication and loss model (GDL). Note that GDL is prevalent in evolutionary histories and is a part of the powerful duplication-loss-coalescence evolutionary model (DLCoal) by Rasmussen and Kellis. In this work we prove that ASTRAL is statistically consistent under the general DLCoal model. Therefore, our result supports the empirical evidence from the simulation-based studies. More broadly, we prove that a randomly chosen quartet from a gene tree (with unique taxa) is more likely to agree with the respective species tree quartet than any of the two other quartets.
\end{abstract}


%\end{frontmatter}

\section{Introduction}
Accurate inference of evolutionary histories of species is one of the most challenging tasks in biology. One of the most hindering factors in the field is the difficulty of proper evaluation of computed species phylogenies. This is due to the fact that researchers rarely know the \emph{true} evolutionary history~\cite{Bininda-Emonds:2004}.

Therefore, the common strategy in the phylogenetic community is to rely on the established statistical models of evolution. One of the most prominent such models is the multispecies coalescent model~\cite{Rannala:2003MSC} that accounts for incomplete lineage sorting (ILS) also known as deep coalescence. ILS is a prevalent factor that causes discordance between the observed gene tree topologies and the host species tree~\cite{Allman:2018splits-MSC}. In fact, a large body of work in phylogenetics is dedicated to the design of species tree inference methods that are statistically consistent under MSC. Statistical consistency implies that as the number of observed gene trees grows, the species tree estimate converges to the true species tree that ``generated'' the observed data. Multiple phylogenetic inference methods have been demonstrated to be statistically consistent, cf. ASTRAL~\cite{Zhang:2018astral3}, NJst~\cite{Liu:2011NJst}, ASTRID~\cite{Vachaspati:2015astrid}, STAR~\cite{Liu:2009STAR}, STEM~\cite{Kubatko:2009STEM}, MP-EST~\cite{Liu:2010MP-EST}, BUCKy~\cite{Larget:2010BUCKy}, GLASS~\cite{Mossel:2008GLASS}, and others.

In the recent years ASTRAL became one of the most popular species tree inference methods by practitioners. Note that ASTRAL's objective function is built on the notion of \emph{quartets} (see Figure~\ref{fig:quartets}). In particular, the proof that ASTRAL is statistically consistent under MSC stems from two observations. First, Allman et al.~\cite{Allman:2011unrooted} demonstrated that if a species tree displays a quartet $q$ then $q$ is also the most likely observed (unrooted) gene tree topology. Second, it can be seen that every species tree clade will eventually appear in at least one of the observed gene trees.

More recently Legried et al. proved that a version of ASTRAL extended to work with multi-locus gene trees, called \emph{ASTRAL-one}, is statistically consistent under the gene duplication and loss model (GDL)~\cite{Legried:2019dl}. Note that GDL is a part of the well-recognized unified duplication-loss-coalescence (DLCoal) model of gene tree evolution by Rasmussen and Kellis~\cite{Rasmussen:2012unified}. DLCoal simultaneously accounts for three crucial types of evolutionary events that shape gene evolution. Namely, duplications, losses, and incomplete lineage sorting events. The DLCoal process involves two steps, (i) a birth/death process within the branches of the species tree creates a \emph{locus tree} (i.e., the GDL process) and (ii) a bounded multispecies coalescence process acting on the locus tree generates the observed \emph{gene tree}. See Figure~\ref{fig:dlc-example} for an example.

In this work we prove that ASTRAL-one is statistically consistent under the general DLCoal model. First, we derive gene tree probabilities (constrained to quartets) under the bounded multispecies coalescent model and draw core observations from that analysis. Second, we build on an idea from Legried et al. to systematically separate different duplication-loss scenarios. Then for each such scenario we prove that a random quartet from the gene tree is more likely to agree with the species tree quartet rather than any of the two other quartets.

This result provides a theoretical justification to the findings in~\cite{Du:2019mlocus,Zhang:2019astral-pro} which showcased the accuracy of ASTRAL in presence of duplications, losses, and incomplete lineage sorting.

%One of the most prominent such models is the unified duplication-loss-coalescence (DLCoal) model of gene tree evolution by Rasmussen and Kellis~\cite{Rasmussen:2012unified}. DLCoal simultaneously accounts for three crucial types of evolutionary events that shape gene evolution. Namely, duplications, losses, and deep coalescence (also known as incomplete lineage sorting) events. The DLCoal process involves two steps, (i) a birth/death process within the branches of the species tree creates a \emph{locus tree} and (ii) a bounded multispecies coalescence process acting on the locus tree generates the observed \emph{gene tree}.

\section{Preliminaries}
We denote a rooted \emph{(phylogenetic) tree} by $(T, \omega)$. Here $T$ is the \emph{tree topology} and is a binary rooted tree with the designated root vertex, $\rho(T)$, of degree 2, all internal nodes of degree three, and with leaves that are bijectively labeled by elements of set $\Le(T)$. For convenience, we identify leaves with their labels. Further, tree topologies are \emph{planted} implying that an additional \emph{root edge} is attached to the root vertex.
Then, $\omega$ specifies the lengths of edges in $T$ in \emph{coalescent units} (i.e., the number of generations normalized by the effective population size~\cite{Allman:2011unrooted}). More formally, $\omega : E(T) \to \Q^{+}$. In particular, we assume that all edge lengths are strictly positive.

%Two $X$-trees are identical if there exists a label-preserving graph isomorphism between them.
 %defined as the binary rooted tree such that the minimal connected subgraph of $T$ which contains all leaves from $Y$ is a subdivision of $T|_Y$. %For convenience, we define the \emph{size} of a tree as $|T| := |\Le(T)| = |X|$.

An \emph{unrooted (phylogenetic) tree topology} $T$ is similar to the rooted tree topology, but without a designated root and the root edge. That is, in unrooted tree $T$ all non-leaf vertices have degree 3.

We say that an edge $e$ is \emph{external} if it is incident with a leaf vertex, and otherwise we call $e$ \emph{internal}.
Further, given a set $Y \subset X$, tree topology $T|_Y$ is obtained from $T$ by restricting the leaf-set to $Y$.

A rooted topology $T$ defines a partial order on its nodes: given two nodes $x$ and $y$ we say $x \preceq y$ if $x$ is a descendant of $y$ (and $x \prec y$ if additionally $x \ne y$). %Further, we say that $x$ and $y$ are \emph{incomparable} if neither $x \preceq y$ nor $y \preceq x$. 
For a set $Z \subseteq X$ the \emph{least common ancestor (lca)} of $Z$, denoted $\lca_T(Z)$, is the lowest node $v$ such that each $l \in Z$ is a descendant of $v$.

\begin{figure}
	\centering
	\begin{tikzpicture}[scale=0.65]
	\tikzstyle{vertex} = [circle,draw,fill,inner sep=0pt, minimum size=4pt]
	\tikzstyle{edge} = [draw,ultra thick,-]
	\tikzstyle{hedge} = [draw,ultra thick,-,color=blue]
	\tikzset{vlabel/.style 2 args={#1=1pt of #2}}

	\def\offset{0}
	\foreach \x/\y/\name in {{0/0/v1},{1/0/v2},{-0.7/-0.7/l2},{-0.7/0.7/l1},{1.7/0.7/l3},{1.7/-0.7/l4}} {
		\coordinate [] (\name) at (\offset + \x,\y) {};
	}
	\foreach \v/\u/\param in {{v1/v2/edge},{v1/l2/edge},{v1/l1/edge},{v2/l3/edge},{v2/l4/edge}} {
		\path [\param] (\v) -- (\u);
	}
	\foreach \name/\label/\pos in {l1/a/left,l2/b/left,l3/c/right,l4/d/right} {
		\node [vertex] at (\name) {};
		\node [vlabel={\pos}{\name}] {\label};
	}
	%\node[] at (\offset+0.5,1.5) {(1)};
	
	\def\offset{6}
	\foreach \x/\y/\name in {{0/0/v1},{1/0/v2},{-0.7/-0.7/l2},{-0.7/0.7/l1},{1.7/0.7/l3},{1.7/-0.7/l4}} {
		\coordinate [] (\name) at (\offset + \x,\y) {};
	}
	\foreach \v/\u/\param in {{v1/v2/edge},{v1/l2/edge},{v1/l1/edge},{v2/l3/edge},{v2/l4/edge}} {
		\path [\param] (\v) -- (\u);
	}
	\foreach \name/\label/\pos in {l1/a/left,l2/c/left,l3/b/right,l4/d/right} {
		\node [vertex] at (\name) {};
		\node [vlabel={\pos}{\name}] {\label};
	}
	%\node[] at (\offset+0.5,1.5) {(2)};

	\def\offset{12}
	\foreach \x/\y/\name in {{0/0/v1},{1/0/v2},{-0.7/-0.7/l2},{-0.7/0.7/l1},{1.7/0.7/l3},{1.7/-0.7/l4}} {
		\coordinate [] (\name) at (\offset + \x,\y) {};
	}
	\foreach \v/\u/\param in {{v1/v2/edge},{v1/l2/edge},{v1/l1/edge},{v2/l3/edge},{v2/l4/edge}} {
		\path [\param] (\v) -- (\u);
	}
	\foreach \name/\label/\pos in {l1/a/left,l2/d/left,l3/b/right,l4/c/right} {
		\node [vertex] at (\name) {};
		\node [vlabel={\pos}{\name}] {\label};
	}
	
\end{tikzpicture}
	\caption{All three possible quartets on $a,b,c,d$ leaves.}
	\label{fig:quartets}
\end{figure}
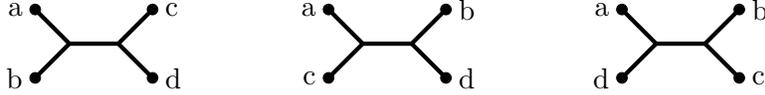

\paragraph{Quartets.} A quartet is an unrooted tree topology with exactly four leaves. Assuming that the leaves are $a, b, c,$ and $d$, we denote the quartets in Figure~\ref{fig:quartets}(left), \ref{fig:quartets}(middle) and  \ref{fig:quartets}(right) as $ab|cd$, $ac|bd$, and $ad|bc$ respectively (based on what pairs of leaves does the internal edge separate).

\subsection{Unified DLCoal model}\label{sec:dlc}
We now review the unified duplication-loss-coalescence (DLCoal) model~\cite{Rasmussen:2012unified}.

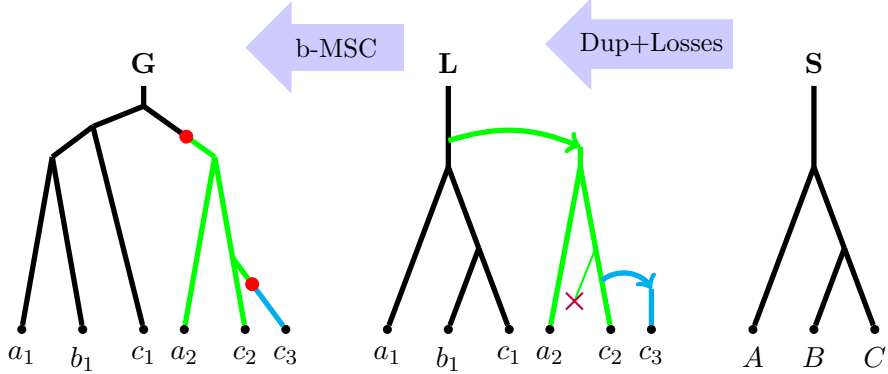
\begin{figure}[!t]
	\centering
	\begin{tikzpicture}[scale=0.27]
	\tikzstyle{vertex} = [circle,draw,fill,inner sep=0pt, minimum size=3pt]
	\tikzstyle{edge} = [draw,line width=2pt,-]
	\tikzset{vlabel/.style 2 args={#1=1pt of #2}}
	\tikzset{cross/.style={cross out, draw=black, minimum size=2*(#1-\pgflinewidth), inner sep=0pt, outer sep=0pt},cross/.default={1pt}}
	\def\step{0.8}
	\def\offset{36}	
	
	\foreach \x/\name\label in {{0/A/$A$},{3/B/$B$},{6/C/$C$}} {
		\node [vertex] (\name) at (\offset+\x,0) {};
		\node [vlabel={below}{\name}] {\label};
	}
	\foreach \x/\y/\name in {{3/12/spr},{3/8/sr},{4.5/4/sbc}} {
		\coordinate [] (\name) at (\offset+\x,\y) {};
	}
	\foreach \u/\v in {spr/sr,sr/sbc,sr/A,sbc/B,sbc/C} {
		\path [edge] (\u) -- (\v);
	}
	\node [vlabel={above}{spr}] {$\mathbf{S}$}; 
	
	%------------------
	\def\offset{18}	
	\foreach \x/\name\label in {{0/l1a/$a_1$},{3/l1b/$b_1$},{6/l1c/$c_1$},8/l2a/$a_2$,11/l2c/$c_2$,13/l3c/$c_3$} {
		\node [vertex] (\name) at (\offset+\x,0) {};
		\node [vlabel={below}{\name}] {\label};
	}
	\foreach \x/\y/\name in {{3/12/lpr},{3/8/l1r},{4.5/4/l1bc},9.5/9/l2pr,9.5/8/l2r,10.25/4/l2bc,13/2/l3r} {
		\coordinate [] (\name) at (\offset+\x,\y) {};
	}
	\foreach \u/\v in {lpr/l1r,l1r/l1bc,l1r/l1a,l1bc/l1b,l1bc/l1c} {
		\path [edge] (\u) -- (\v);
	}
	\foreach \u/\v in {l2pr/l2r,l2r/l2a,l2r/l2c} {
		\path [edge,color=green] (\u) -- (\v);
	}

	\foreach \u/\v in {l3r/l3c} {
		\path [edge,color=cyan] (\u) -- (\v);
	}
	\coordinate [] (ol2) at (\offset+3,9.3) {};
	\coordinate [] (ol3) at (\offset+11 - 0.3*1.5,0.3*8) {};
	\path[edge,color=cyan,->,bend left=20] (ol3) to [out=50,in=130] (l3r); 
	\path[edge,color=green,->,bend left=20] (ol2) to [out=20,in=160] (l2pr); 
	
	\coordinate[below left=20pt and 8pt of l2bc] (l2b) {};
	\path[draw,thick,color=green] (l2bc) -- (l2b);
	\node[cross=4pt,purple,thick] at (l2b) {};
	
	\node [vlabel={above}{lpr}] {$\mathbf{L}$}; 
	
	%------------------
	\def\offset{0}
	\foreach \x/\name\label in {{0/g1a/$a_1$},{3/g1b/$b_1$},{6/g1c/$c_1$},8/g2a/$a_2$,11/g2c/$c_2$,13/g3c/$c_3$} {
		\node [vertex] (\name) at (\offset+\x,0) {};
		\node [vlabel={below}{\name}] {\label};
	}
	\foreach \x/\y/\name in {{6/12/gpr},{6/11/gr},{3.5/10/g1r},1.5/8.5/g1ab,9.5/8.5/g2r,13/2/g3r} {
		\coordinate [] (\name) at (\offset+\x,\y) {};
	}
	\coordinate [] (g3r) at (\offset+11-0.45*1.5,0.45*8) {};
	\path[] (gr) -- (g2r) node [pos=0.6,inner sep=0] (g2pr) {};
	\path[] (g3r) -- (g3c) node [pos=0.4,inner sep=0pt] (g3pr) {};
	\foreach \u/\v in {gpr/gr,gr/g1r,g1r/g1c,g1r/g1ab,g1ab/g1a,g1ab/g1b,gr/g2pr} {
		\path [edge] (\u) -- (\v);
	}
	\foreach \u/\v in {g2pr/g2r,g2r/g2a,g2r/g2c,g3r/g3pr} {
		\path [edge,color=green] (\u) -- (\v);
	}
	
	\path[edge,color=cyan] (g3pr) -- (g3c); 
	
	\draw [fill,color=red] (g2pr) circle (9pt);
	\draw [fill,color=red] (g3pr) circle (9pt);
	
	\node [vlabel={above}{gpr}] {$\mathbf{G}$}; 
	
	\node [single arrow, fill=blue!20,shape border rotate=180, above=3.4cm of l3c, minimum width=35pt] {\small{Dup+Losses}};
	
	\node [single arrow, fill=blue!20,shape border rotate=180, above right=3.1cm and 0.02cm of g3c,minimum width=35pt,minimum height=60pt] {\small{b-MSC}};
	
%	\node[draw,star,star points=7,above right=0pt and 6pt of gr,inner sep=2.5pt,ultra thick,fill=yellow] {};
%	\node[draw,star,star points=7,above right=0pt and 6pt of g3r,inner sep=2.5pt,ultra thick,fill=yellow] {};
		
\end{tikzpicture}
	\caption{An example of a gene tree $G$, locus tree $L$, and species tree $S$. Note that the arrows in the locus tree represent the duplication events, and the cross represents a loss event. Further, the red circles on the gene tree represent the duplication points.}
	\label{fig:dlc-example}
\end{figure}

\paragraph{Species tree.}
A species tree $(T_S, \omega_S)$ represents an evolutionary history of species. Leaves of $T_S$ are labeled by the extant species names.

\paragraph{Locus tree.} A locus tree $(T_L, \omega_L)$ represents a duplication/loss history of a fixed gene. A locus tree is obtained from the species tree by running the duplication-loss process~\cite{Rasmussen:2012unified,Legried:2019dl} on top of it. More specifically, The process starts in the root edge of the species tree with exactly one locus. Then, developing along the branches of the species tree, every locus (independently) has an exponential probability of being duplicated (i.e., a new locus is created) or being lost. See Figure~\ref{fig:dlc-example} for an example.

Locus tree leaves are labeled by gene names.

\paragraph{Gene tree.} A gene tree $(T_G, \omega_G)$ represents a gene evolutionary history. The gene tree is obtained from the locus tree by running the \emph{bounded multispecies coalescent (b-MSC)} process on top of it~\cite{Rasmussen:2012unified} (see Section~\ref{sec:c-ils} for a more detailed description of that process). Figure~\ref{fig:dlc-example} provides an example of that process.

%\subsubsection{DL process}

\subsection{Multispecies coalescent (MSC) model}\label{sec:ils}
In the standard multispecies coalescent model~\cite{Rannala:2003MSC} gene lineages are followed backwards in time (from the leaves to the root).

For simplicity, we assume that there is exactly one gene lineage starting in every extant locus tree leaf. If two or more lineages enter the same locus tree edge, then an exponential distribution defines a probability of a specific pattern of coalescence of those lineages.

In particular, for any two lineages $a,b$ that entered the same edge the probability that they coalesce within time $x$ (specified in terms of coalescent units) is as follows:
\[
P[a,b \text{ coalesced within time } x] = 1 - e^{-x}.
\]
%Here we adopt a notation of the form $T_x$ that 

More generally, we denote the probability that $i$ lineages coalesce into $j$ lineages within time $x$ ($j \le i$) by 
$
g_{i,j}(x).
$ 
This value can be computed using the following formula~\cite{Allman:2011unrooted}:
\[
g_{i,j}(x) = \sum_{k = j}^{i}\exp\left(-{k \choose 2}x\right)\frac{(2k-1)(-1)^{k-j}}{j!(k-j)!(j+k-1)}\prod_{m=0}^{k-1}\frac{(j+m)(i-m)}{i+m}.
\]

Further, it is important to note that if more than 2 lineages enter the same edge, then there is an equal probability for any pair of the lineages to coalesce first (i.e., the process is symmetric).

\subsection{Bounded MSC (b-MSC) model}\label{sec:c-ils}
The constraints on MSC in the unified DLCoal model appear due to the duplication points. In particular, assume that a duplication occurred at time point $d$ in time (counting backwards from leaves). Further, let $a$ and $b$ be locus leaves that are located below that duplication (i.e., $a$ and $b$ exist as a consequence of that duplication). Then we know that lineages $a$ and $b$ must coalesce prior to time point $d$. Therefore, the probability that any two lineages $a,b$, which entered the same edge below a duplication point $p$, coalesce within time $x$ is as follows:
\[
P[a,b \text{ coalesced within time } x \mid a,b \text{ coalesced prior to } p] = \frac{1 - e^{x}}{P[a,b \text{ coalesced prior to } p]}.
\]

Where $P[a,b \text{ coalesced prior to } p]$ is determined by the original, unbounded MSC model.

\section{Quartet probabilities under b-MSC}\label{sec:probs}
To obtain our main result we need to compute the probabilities of each quartet appearing in the gene tree tree based on a fixed locus tree topology. Note that Allman et al.~\cite{Allman:2011unrooted} explicitly computed these probabilities in the case of \emph{unconstrained MSC}. In our case we need to incorporate cases, when duplications (locus creation events) appear along the edges of the locus tree.

\noindent\textbf{Remark:} From now on, for convenience, we only focus on trees with exactly four leaves.

\smallskip
Without loss of generality assume that the locus tree $L$ displays the quartet $\qab$. Then there are two cases: either (i) $L$ is a balanced rooted tree or (ii) $L$ is a caterpillar tree. We now explore both those cases.

\subsection{$L$ is balanced}\label{sec:prob-balanced}

\begin{figure}[!h]
	\centering
	\begin{tikzpicture}[scale=0.55]
	\def\offset{0}
	\input{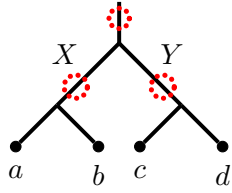}
	\path[] (x) -- (r) node [midway, above left] {$X$} node [pos=0.3] (mx) {};
	\path[] (y) -- (r) node [midway, above right] {$Y$} node [pos=0.3] (my) {};
	\path[] (r) -- (pr) node [pos=0.6] (mr) {};
	
	\draw [line width=2pt, line cap=round, dash pattern=on 0pt off 1.5*\pgflinewidth,color=red] (mx) circle (8pt);
	\draw [line width=2pt, line cap=round, dash pattern=on 0pt off 1.5*\pgflinewidth,color=red] (my) circle (8pt);
	\draw [line width=2pt, line cap=round, dash pattern=on 0pt off 1.5*\pgflinewidth,color=red] (mr) circle (7pt);
\end{tikzpicture}
	\caption{The balanced quartet representing the locus tree and displaying quartet $\qab$. The dotted circles indicate potential duplication locations that can affect gene tree probabilities.}
	\label{fig:q-balanced}
\end{figure}

%For convenience, we set $x := \omega_L(X), y := \omega_L(Y)$ to be the lengths of edges $X$ and $Y$ respectively.

\subsubsection{Duplications along the $X$ or $Y$ edges.}\label{sec:balanced-XYdup}
Assume that a duplication has occurred along the $X$ and/or $Y$ edge. Recall that a duplication point indicates that gene lineages below it in the locus tree must coalesce prior to the duplication (when looking backwards in time). Therefore, if there is a duplication along the $X$ edge, then lineages corresponding to genes $a$ and $b$ must coalesce on that edge. That is, the gene tree must display quartet $\qab$. Similarly, the same is true if a duplication is located on the $Y$ edge. Hence,

\begin{tabular}{cc}
\begin{minipage}{0.6\textwidth}
	\begin{gather*}
	P[\qab \in G] = 1,\\
	P[\qac \in G] = P[\qad \in G] = 0.
	\end{gather*}
\end{minipage}&
\begin{minipage}{0.3\textwidth}
	\begin{tikzpicture}[scale=0.5]
	\def\offset{0}
	\input{graphics/qbalanced-base}
	\path[] (x) -- (r) node [midway, above left] {$X$} node [pos=0.3] (mx) {};
	\path[] (y) -- (r) node [midway, above right] {$Y$} node [pos=0.3] (my) {};
	\path[] (r) -- (pr) node [pos=0.6] (mr) {};
	
	\draw [fill,color=red] (mx) circle (8pt);
	\draw [line width=2pt, line cap=round, dash pattern=on 0pt off 1.5*\pgflinewidth,color=red] (my) circle (8pt);
	\draw [line width=2pt, line cap=round, dash pattern=on 0pt off 1.5*\pgflinewidth,color=red] (mr) circle (7pt);
\end{tikzpicture}
\end{minipage}
\end{tabular}

\subsubsection{General case}
We now demonstrate that for balanced $L$ there always exists a duplication along either the $X$ or $Y$ edge. More formally, see Lemma~\ref{lem:balanced-lt}.
\begin{lemma}\label{lem:balanced-lt}
	Let $L$ be a balanced locus tree displaying a quartet $q$. Then $P[q \in G \mid L] = 1$.
\end{lemma}
\begin{proof}
	Let $r$ be the root of $L$. Note that every vertex in the locus tree corresponds to a locus creation event. This implies that one of the child edges of $r$ must have a duplication on it. Then by the above analysis in Section~\ref{sec:balanced-XYdup}, gene tree $G$ must display the same quartet $q$.
\end{proof}
\paragraph{Remark:} Note that potential duplications along the external edges do not affect the coalescence process.

\subsection{$L$ is a caterpillar}\label{sec:prob-cat}
%We now similarly explore the 4-leaf caterpillar topology and respective quartet probabilities.

\begin{figure}[!h]
	\centering
	\begin{tikzpicture}[scale=0.6]
	\def\offset{0}
	\input{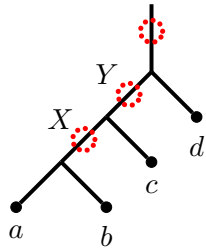}
	\path[] (x) -- (y) node [midway, above left] {$X$} node [pos=0.5] (mx) {};
	\path[] (y) -- (r) node [midway, above left] {$Y$} node [pos=0.5] (my) {};
	\path[] (r) -- (pr) node [pos=0.6] (mr) {};
	
	\draw [line width=2pt, line cap=round, dash pattern=on 0pt off 1.5*\pgflinewidth,color=red] (mx) circle (7pt);
	\draw [line width=2pt, line cap=round, dash pattern=on 0pt off 1.5*\pgflinewidth,color=red] (my) circle (7pt);
	\draw [line width=2pt, line cap=round, dash pattern=on 0pt off 1.5*\pgflinewidth,color=red] (mr) circle (7pt);
\end{tikzpicture}
	\caption{The caterpillar quartet representing the locus tree and displaying quartet $\qab$.}
	\label{fig:q-cat}
\end{figure}

For convenience, we set $x := \omega_L(X), y := \omega_L(Y)$ to be the lengths of edges $X$ and $Y$ respectively.
We now explore all possibilities of duplication placements on edges of $L$.

\subsubsection{No duplications (unbounded MSC).}\label{sec:cat-nodup} In this case the quartet probabilities are given by Allman et al.~\cite{Allman:2011unrooted}. In particular,

\begin{tabular}{cc}
\begin{minipage}{0.6\textwidth}
	\begin{gather*}
	P[\qab \in G] = 1 -\frac{2}{3}e^{-x},\\
	P[\qac \in G] = P[\qad \in G] = \frac{1}{3}e^{-x}.
	\end{gather*}
\end{minipage}&
\begin{minipage}{0.3\textwidth}
	\begin{tikzpicture}[scale=0.55]
	\def\offset{0}
	\tikzstyle{vertex} = [circle,draw,fill,inner sep=0pt, minimum size=4pt]
\tikzstyle{edge} = [draw,line width=1.5pt,-]
\tikzstyle{hedge} = [draw,ultra thick,-,color=blue]
\tikzset{vlabel/.style 2 args={#1=1pt of #2}}

\foreach \x/\y/\name/\label in {0/0/a/$a$,2/0/b/$b$,3/1/c/$c$,4/2/d/$d$} {
 	\node [vertex] (\name) at (\offset+\x,\y) {};
 	\node [vlabel={below}{\name}] {\label};
}

\foreach \x/\y/\name in {1/1/x,2/2/y,3/3/r,3/4.5/pr} {
	\coordinate (\name) at (\offset+\x,\y) {};
}

\foreach \u/\v/\label/\pos in {x/a//,x/b//,y/x/X/left,r/y/Y/right,pr/r//,y/c//,r/d//} {
	\path [edge] (\u) -- (\v);
}
	\path[] (x) -- (y) node [midway, above left] {$X$} node [pos=0.5] (mx) {};
	\path[] (y) -- (r) node [midway, above left] {$Y$} node [pos=0.5] (my) {};
	%\path[edge] (r) -- ($(r)+(0,2)$) node [pos=0.75] (mr) {};
\end{tikzpicture}
\end{minipage}
\end{tabular}

\subsubsection{$X$ edge duplication.}\label{sec:cat-Xdup}
Similarly to the balanced case, it is not difficult to see that

\begin{tabular}{cc}
	\begin{minipage}{0.6\textwidth}
		\begin{gather*}
		P[\qab \in G] = 1,\\
		P[\qac \in G] = P[\qad \in G] = 0.
		\end{gather*}
	\end{minipage}&
	\begin{minipage}{0.3\textwidth}
		\begin{tikzpicture}[scale=0.55]
	\def\offset{0}
	\input{graphics/qcat-base}
	\path[] (x) -- (y) node [midway, above left] {$X$} node [pos=0.5] (mx) {};
	\path[] (y) -- (r) node [midway, above left] {$Y$} node [pos=0.5] (my) {};
	\path[] (r) -- (pr) node [pos=0.6] (mr) {};
	
	\draw [fill,color=red] (mx) circle (7pt);
	\draw [line width=2pt, line cap=round, dash pattern=on 0pt off 1.5*\pgflinewidth,color=red] (my) circle (7pt);
	\draw [line width=2pt, line cap=round, dash pattern=on 0pt off 1.5*\pgflinewidth,color=red] (mr) circle (7pt);
\end{tikzpicture}
	\end{minipage}
\end{tabular}

\subsubsection{$Y$ edge duplication.}\label{sec:cat-Ydup}
Assume that a duplication occurred on the $Y$ edge as shown in the figure below and no duplications occurred along on the $X$ edge. Then

\noindent
\begin{minipage}{0.7\textwidth}
	\flushleft
	\begin{gather*}
	\begin{split}
	&P[\qab \in G] = P[\qab \in G \mid a,b,c \text{ coalesced before duplication}] \\
	&= 1 - P\left[
		\begin{split}
		&a,b \text{ did not coalesce on } X;\\
		&\qac \text{ or } \qad \text{ obtained during time } l\ \ 
		\end{split}
	\Bigg| a,b,c \text{ coalesced before dup.}\right]\\
	&= 1 - \frac{2}{3}e^{-x}\frac{g_{3,1}(l)}{P[a,b,c \text{ coalesced before duplication}]}
	\end{split}\\
	P[\qab \in G] = P[\qad \in G]
	= \frac{1}{3}e^{-x}\frac{g_{3,1}(l)}{P[a,b,c \text{ coalesced before duplication}]}
	\end{gather*}
\end{minipage}%
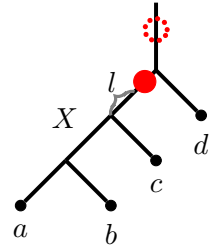
\begin{minipage}{0.3\textwidth}
	\flushright
	\begin{tikzpicture}[scale=0.6]
	\def\offset{0}
	\input{graphics/qcat-base}
	\path[] (x) -- (y) node [midway, above left] {$X$} node [pos=0.5] (mx) {};
	\path[] (y) -- (r) node [pos=0.75] (my) {};
	\path[] (r) -- (pr) node [pos=0.6] (mr) {};
	
	%\draw [fill,color=red] (mx) circle (7pt);
	\draw [fill,color=red] (my) circle (7pt);
	\draw [line width=2pt, line cap=round, dash pattern=on 0pt off 1.5*\pgflinewidth,color=red] (mr) circle (7pt);
	\draw [decorate,decoration={brace,amplitude=4pt},xshift=5pt,yshift=0pt,very thick,color=gray]
	(y) -- (my) node [midway, above left, xshift=2pt,color=black] {$l$};
\end{tikzpicture}
\end{minipage}

\subsubsection{Root edge duplication.}\label{sec:cat-rootdup}
Assume that a duplication occurred on the root edge as shown in the figure below and no duplications occurred along the $X$ and $Y$ edges.

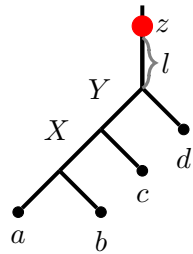
\begin{center}
	\begin{tikzpicture}[scale=0.55]
	\def\offset{0}
	\input{graphics/qcat-base}
	\path[] (x) -- (y) node [midway, above left] {$X$} node [pos=0.5] (mx) {};
	\path[] (y) -- (r) node [midway, above left] {$Y$} node [pos=0.5] (my) {};
	\path[edge] (r) -- ($(r)+(0,2)$) node [pos=0.75] (mr) {};
	
	%\draw [fill,color=red] (mx) circle (8pt);
%	\draw [line width=2pt, line cap=round, dash pattern=on 0pt off 1.5*\pgflinewidth,color=red] (my) circle (8pt);
	\draw [fill,color=red] (mr) circle (7pt);
	\node[right, xshift=1pt] at (mr) {$z$};
	\draw [decorate,decoration={brace,amplitude=5pt,mirror},xshift=5pt,yshift=0pt,very thick,color=gray]
	(r) -- (mr) node [midway, right, xshift=3pt,color=black] {$l$};
\end{tikzpicture}
\end{center}

We start with computing the probability of the $\qac$ quartet. %\todo{Specify $g_{i,j}(x)$ meaning earlier}
\[
\begin{split}
P[\qac \in G] &= \frac{P\left[
	\begin{split}
	&a,b \text{ did not coalesce on } X;\\
	&a,c \text{ coalesced on } Y \text{ first};\\
	&\text{remaining lineages coalesced before } z
	\end{split}
\right]+
P\left[
	\begin{split}
	&a,b \text{ did not coalesce on } X;\\
	&\text{no coalescence on } Y;\\
	&\qac \text{ obtained within time } l
	\end{split}
\right]
}
{P[a,b,c,d \text{ coalesced before duplication}]}\\
&= \frac{\frac{1}{3}e^{-x}
	\big(
	g_{3,2}(y)g_{3,1}(l) + g_{3,1}(y)g_{2,1}(l) + g_{3,3}(y)g_{4,1}(l)
	\big)
}
{P[a,b,c,d \text{ coalesced before duplication}]}.
\end{split}
\]
Further, by symmetry $P[\qac \in G] = P[\qad \in G]$. Therefore,

\[
P[\qab \in G] = 1 - \frac{2}{3}e^{-x}
\left(
	\frac{g_{3,2}(y)g_{3,1}(l) + g_{3,1}(y)g_{2,1}(l) + g_{3,3}(y)g_{4,1}(l)}
	{P[a,b,c,d \text{ coalesced before duplication}]}
\right).
\]

\subsection{Core observations}
It is not difficult to see from the above derivations that for a fixed locus tree topology that displays $\qab$ (balanced or caterpillar), if one increases the length of edge $X$ then the probability $P[\qab \in G]$ grows.
More formally, see Observation~\ref{obs:grow}.
\begin{observ}\label{obs:grow} $ $\\ 
	\begin{enumerate}
		\item[(i)] Let $L_1$ and $L_2$ be two \emph{balanced} trees displaying $\qab$ (as in Fig.~\ref{fig:q-balanced}), with identical locations of duplication nodes and with $\omega_{L_1}(X) > \omega_{L_2}(X)$ and $\omega_{L_1}(Y) \ge \omega_{L_2}(Y)$. Then
		\[
		P[\qab \in G \mid L_1] > P[\qab \in G \mid L_2].
		\]
		\item[(ii)] Let $L_1$ and $L_2$ be two \emph{caterpillar} trees displaying $\qab$ (as in Fig.~\ref{fig:q-cat}), with identical locations of duplication nodes and with $\omega_{L_1}(X) > \omega_{L_2}(X)$ and $\omega_{L_1}(Y) \ge \omega_{L_2}(Y)$. Then
		\[
		P[\qab \in G \mid L_1] > P[\qab \in G \mid L_2].
		\]
	\end{enumerate}
\end{observ}

Further, from the above derivations we observe the following.
\begin{lemma}\label{lem:greater}
	For a locus tree $L$ that displays $\qab$ (regardless of duplication locations) we have $P[\qab \in G \mid L] > P[\qac \in G \mid L] = P[\qad \in G \mid L]$.
\end{lemma}
\begin{proof}
	Observe that this statement is not trivial only in two cases. (i) $L$ is a caterpillar and the lowest duplication is on $Y$ edge. (ii) $L$ is a caterpillar and the lowest duplication is at the root edge.
	
	In case (i) it is sufficient to show that $P[a,b,c \text{ coalesced before duplication}] \ge g_{3,1}(l)$. Indeed observe that
	\[
		P[a,b,c \text{ coalesced before duplication}] = \sum_{k = 2}^{3}g_{k,1}(l)P[k \text{ lineages entered edge} Y] \ge g_{3,1}(l).
	\]
	The inequality holds since $g_{k,1}(l) \ge g_{3,1}(l)$ for all $k \in \set{2,3}$.
	
	Note that the proof in the other case is similar.\todo{AM: double check case (ii), as it is most challenging.}
\end{proof}

\section{Consistency of ASTRAL-one}
We now prove that ASTRAL-one method is statistically consistent under the DLCoal model.

\begin{theorem}
	Let $S$ be a fixed species tree and let $\cG$ be a collection of gene trees that independently evolved within $S$ according to the DLCoal process. Then as the number of trees in $\cG$ goes to infinity, the probability that $\hat{S}$, the species tree estimate by ASTRAL-one, is equal to $S$ goes to 1.
\end{theorem}

For this result it is sufficient (see~\cite{Legried:2019dl}) to prove the following:

\begin{theorem}\label{thm:main}
	Let $S$ be a species tree with four leaves that displays quartet $AB|CD$, and let $G$ be a gene tree that evolved in $S$ according to the DLCoal process. If one picks genes $a, b, c, d$ (that correspond to species $A, B, C, D$ respectively) uniformly at random (assuming they exist) from $G$, then $P[\qab \in G] > P[\qac \in G] = P[\qad \in G]$.
\end{theorem}
Therefore, the remainder of the section is dedicated to the proof of Theorem~\ref{thm:main}.
We first prove the theorem for $S$ being balanced, and then prove it for the caterpillar case.

\begin{figure}
	\centering
	\begin{tikzpicture}[scale=0.45]
	\tikzstyle{vertex} = [circle,draw,fill,inner sep=0pt, minimum size=4pt]
	\tikzstyle{edge} = [draw,line width=1.5pt,-]
	\tikzstyle{ledge} = [draw,line width=1.5pt,-,color=blue]
	\tikzstyle{hedge} = [draw,ultra thick,-,color=blue]
	\tikzset{vlabel/.style 2 args={#1=1pt of #2}}
	\def\offset{0}
	
	%\foreach \x/\name/\label in {0/a/$a$,2/b/$b$,3/c/$c$,5/d/$d$} {
	% 	\node [vertex] (\name) at (\offset+\x,0) {};
	% 	\node [vlabel={below}{\name}] {\label};
	%}
	
	\foreach \x/\y/\name in {0/0/llr,0/3/ltr,6/0/rlr,6/3/rtr,3/3/lr,-5/-5/la,-3/-5/ra,0/-5/lb,2/-5/rb,4/-5/lc,6/-5/rc,9/-5/ld,11/-5/rd,-1.5/-3.5/abs,7.5/-3.5/cds,0.5/-3.5/rabs,5.5/-3.5/lcds,3/-1/abcds} {
		\coordinate (\name) at (\offset+\x,\y) {};
	}
	
	\foreach \u/\v in {ltr/llr,rtr/rlr,llr/la,la/ra,ra/abs,abs/lb,lb/rb,rb/rabs,rabs/abcds,abcds/lcds,lcds/lc,lc/rc,rc/cds,cds/ld,ld/rd,rd/rlr} {
		\path [edge] (\u) -- (\v);
	}
	\draw [ultra thick,dashed] (llr) -- (rlr);
	\foreach \x/\slabel in {-4/$A$,1/$B$,5/$C$,10/$D$} {
		\node[] at (\offset+\x+0.3,-6.2) {\slabel};
	}

	\foreach \x/\name in {0.5/rl1,1.5/rl2,2.5/rl3,3.8/rl4,4.9/rl5} {
		\coordinate (b\name) at (\offset+\x,0) {};
		\coordinate (t\name) at (\offset+\x,0.5) {};
		\path [ledge] (b\name) -- (t\name); 
		\path [ledge,dotted] (lr) -- (t\name);
	}
	\foreach \x/\name/\label in {-4.5/a/$a$,0.5/b/$b$,4.5/c/$c$,9.5/d/$d$} {
		\coordinate [] (\name) at (\offset+\x,-5) {};
		\node [color=blue, below=0pt of \name] {\label};
	}
	\path [ledge] (a) to[out=40,in=-100] (brl1);
	\path [ledge] (b) to[out=140,in=-120] (brl3);
	\coordinate (lcd) at (\offset+6.5,-2) {};
	\path [ledge] (c) to[out=35,in=-95] (lcd);
	\path [ledge] (d) to[out=95,in=-60] (lcd);
	\path [ledge] (lcd) -- (brl5);
\end{tikzpicture}
	\caption{An example of the partial embedding of a locus tree into balanced $S$. The blue lineages correspond to the locus tree. Note that the five locus lineages crossing the dashed speciation line are \emph{root lineages}.}
	\label{fig:embed-balanced}
\end{figure}
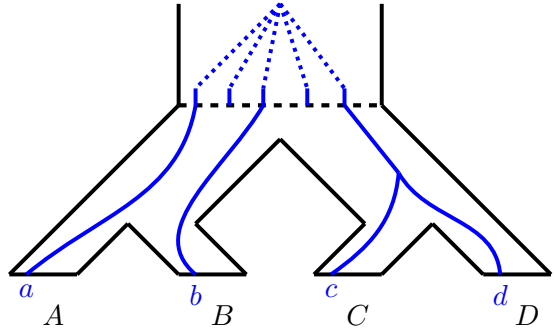

\subsection{$S$ is balanced}\label{sec:consistency-balanced}
\todo{make sure proper $>$ is achieved in at least one of the cases.}
Similarly to~\cite{Legried:2019dl}, we first of all implicitly condition our probability space on the event that at least one of each $a, b, c$ and $d$ genes must be present in $G$.
Further, we condition our probability space on a fixed number of locus tree lineages existing at the speciation point at the root of $S$. That is, consider the duplication/loss (birth/death) process occurring within the root branch of $S$. Then let $RL$ be the random variable denoting the number of locus lineages at the speciation point (see Figure~\ref{fig:embed-balanced}). We are going to prove that
\[
P[\qab \in G \mid RL=l] > P[\qac \in G \mid RL=l] = P[\qad \in G \mid RL=l]
\]
for any fixed value of $l = \set{1,2,\ldots}$. Therefore, for convenience, we do not explicitly write the condition $RL=l$ in probability equations throughout the rest of the proof. Further, we refer to the set of these $l$ locus lineages as \emph{root lineages}.

Let now $i_a \in \set{1, \ldots, l}$ be the index of a root lineage, from which gene $a$ has descended. Similarly, we define $i_b, i_c$, and $i_d$. For the better readability of the rest of the proof, we introduce the notation to describe scenarios of the type $i_a = i_b = i_c \ne i_d$. In particular, we write $(abc,d)$ for that scenario, we write $(ab,cd)$ to denote the scenario $i_a = i_b \ne i_c = i_d$, and we write $(a,b,c,d)$ to denote the scenario, where all $i_x$ are distinct.

Then by the law of total probability.
\[
P[\qab \in G] = \sum_{I}{P[\qab \in G, I]}.
\]
Where $I$ is one of the above scenarios (i.e., a partition of set $\set{a,b,c,d}$ or a combination of such partitions). In particular $I \in \set{(a,b,c,d);(ab,cd)\lor(ac,bd);(ab,c,d)\lor(cd,a,b)\lor(ac,b,d)\lor (bd,a,c);\allowbreak(abc,d)\lor\allowbreak(abd,c)\lor\allowbreak(acd,b)\lor(bcd,a)\lor(abcd);(ad,bc)\lor(ad,b,c)\lor(bc,a,d)}$. Observe that we cover all possible scenarios/partitions here.

We now prove that $P[\qab \in G] > P[\qac \in G]$. Note that $P[\qad \in G] = P[\qac \in G]$ simply follows from the fact that swapping $c$ and $d$ leaf labels does not affect the probabilities. Let us carry out the proof by considering different values of $I$. That is, we prove that $P[\qab \in G, I] \ge P[\qac \in G, I]$ for all of the above $I$, and at least in one case the \emph{strict} inequality holds.

\subsubsection{Case $\boldsymbol{I = (a,b,c,d)}$} By symmetry of the DLCoal process at the root edge, it is not difficult to see that in this case $P[\qab \in G \mid I] = P[\qac \in G \mid I]$. Hence, $P[\qab \in G, I] = P[\qac \in G, I]$

\subsubsection{Case $\boldsymbol{I = (ab,cd) \lor (ac,bd)}$} \label{sec:(ab,cd)}
We need to show that
\[
\begin{split}
&P[\qab \in G, I] = P[\qab \in G \mid (ab,cd)]P[(ab,cd)] + P[\qab \in G \mid (ac,bd)]P[(ac,bd)]\\ 
&\ge P[\qac \in G \mid (ab,cd)]P[(ab,cd)] + P[\qac \in G \mid (ac,bd)]P[(ac,bd)] = P[\qac \in G, I].
\end{split}
\]

Observe the following.

\begin{figure}
	\centering
	\begin{tikzpicture}[scale=0.45]
	\tikzstyle{vertex} = [circle,draw,fill,inner sep=0pt, minimum size=4pt]
	\tikzstyle{edge} = [draw,line width=1.5pt,-]
	\tikzstyle{ledge} = [draw,line width=1.5pt,-,color=blue]
	\tikzstyle{hedge} = [draw,ultra thick,-,color=blue]
	\tikzset{vlabel/.style 2 args={#1=1pt of #2}}
	\def\offset{0}
	
	\foreach \x/\y/\name in {0/0/llr,0/3/ltr,6/0/rlr,6/3/rtr,3/3/lr,-3/-3/lab,1/-3/rab,5/-3/lcd,9/-3/rcd,3/-1/rs,0.2/-1.5/cab,5.5/-1.5/ccd,1.2/0/brl1,2.8/0/brl2,2/2/irl} {
		\coordinate (\name) at (\offset+\x,\y) {};
	}
	
	\foreach \u/\v in {ltr/llr,rtr/rlr,llr/lab,rlr/rcd,rs/rab,rs/lcd} {
		\path [edge] (\u) -- (\v);
	}
	\draw [ultra thick,dashed] (llr) -- (rlr);
	\path [ledge] (brl1) -- (irl);
	\path [ledge] (brl2) -- (irl);
	\path [ledge,dotted] (lr) -- (irl);
	\foreach \x/\name in {4/rl3,5.5/rl4} {
		\coordinate (b\name) at (\offset+\x,0) {};
		\coordinate (t\name) at (\offset+\x,0.5) {};
		\path [ledge] (b\name) -- (t\name); 
		\path [ledge,dotted] (lr) -- (t\name);
	}
	\foreach \x/\name/\label in {-2.1/a/$a$,0/b/$b$,5.5/c/$c$,7.5/d/$d$} {
		\coordinate [] (\name) at (\offset+\x,-3.2) {};
		\node [color=blue, below=0pt of \name] {\label};
	}
	\path [ledge] (a) to[out=60,in=-130] (cab);
	\path [ledge] (b) to[out=80,in=-100] (cab);
	\path [ledge] (c) to[out=70,in=-95] (ccd);
	\path [ledge] (d) to[out=95,in=-60] (ccd);
	\path [ledge] (cab) -- (brl1);
	\path [ledge] (ccd) -- (brl2);
	
	%--------------
	\def\offset{15}
	
	\foreach \x/\y/\name in {0/0/llr,0/3/ltr,6/0/rlr,6/3/rtr,3/3/lr,-3/-3/lab,1/-3/rab,5/-3/lcd,9/-3/rcd,3/-1/rs,0.2/-1.5/cab,5.5/-1.5/ccd,1.2/0/brl1,2.8/0/brl2,2/2/irl} {
		\coordinate (\name) at (\offset+\x,\y) {};
	}
	
	\foreach \u/\v in {ltr/llr,rtr/rlr,llr/lab,rlr/rcd,rs/rab,rs/lcd} {
		\path [edge] (\u) -- (\v);
	}
	\draw [ultra thick,dashed] (llr) -- (rlr);
	\path [ledge] (brl1) -- (irl);
	\path [ledge] (brl2) -- (irl);
	\path [ledge,dotted] (lr) -- (irl);
	\foreach \x/\name in {4/rl3,5.5/rl4} {
		\coordinate (b\name) at (\offset+\x,0) {};
		\coordinate (t\name) at (\offset+\x,0.5) {};
		\path [ledge] (b\name) -- (t\name); 
		\path [ledge,dotted] (lr) -- (t\name);
	}
	\foreach \x/\name/\label in {-2.1/a/$a$,0/b/$b$,5.5/c/$c$,7.5/d/$d$} {
		\coordinate [] (\name) at (\offset+\x,-3.2) {};
		\node [color=blue, below=0pt of \name] {\label};
	}
	\path [ledge] (a) to[out=60,in=-130] (brl1);
	\path [ledge] (b) to[out=70,in=-110] (brl2);
	\path [ledge] (c) to[out=70,in=-25] (brl1);
	\path [ledge] (d) to[out=100,in=-30] (brl2);
\end{tikzpicture}
	\caption{Left: the embedding of a locus tree $L_{(ab,cd)}$. Right: the embedding of a locus tree $L_{(ac,bd)}$.}
	\label{fig:2lin-embed}
\end{figure}
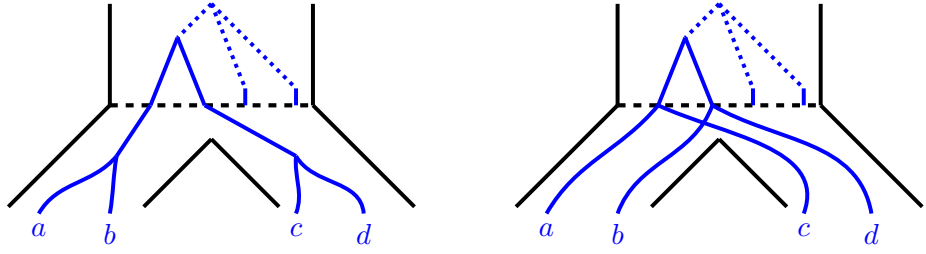

\begin{lemma}\label{lem:2lin}
	$P[\qab \in G \mid (ab,cd)] = P[\qac \in G \mid (ac,bd)] = 1$.
\end{lemma}
\begin{proof}
	Consider the locus trees $L_{(ab,cd)}$ and $L_{(ac,bd)}$ for the $(ab,cd)$ and $(ac,bd)$ cases respectively (see Figure~\ref{fig:2lin-embed}). Note that we only consider the part of the locus tree restricted to the four selected genes $a,b,c,d$. It is not difficult to see that both $L_{(ab,cd)}$ and $L_{(ac,bd)}$ are balanced. Therefore, by Lemma~\ref{lem:balanced-lt}, $P[\qab \in G \mid (ab,cd)] = P[\qac \in G \mid (ac,bd)] = 1$. 
\end{proof}
\begin{corollary}\label{cor:2lin-2}
	$P[\qac \in G \mid (ab,cd)] = P[\qab \in G \mid (ac,bd)] = 0$.
\end{corollary}

\begin{lemma}\label{lem:ab,cd}
	$P[(ab,cd)] \ge P[(ac,bd)]$.
\end{lemma}
\begin{proof}
	Our proof is similar to the proof of Lemma~1 in Legried et al.~\cite{Legried:2019dl}.
	In particular, let $N_i \in \set{0, 1, \ldots}$ be the number of locus lineages that descended from a root lineage $i \in \set{1, \ldots, l}$ and that existed at the moment of speciation into species $A$ and $B$. Similarly, we define $M_i$ variables for the number of lineages that existed at the speciation at the parent of $C$ and $D$.
	
	Observe that $P[(ab,cd)] = P[i_a = i_b, i_c = i_d] - P[(abcd)]$ and $P[(ac,bd)] = P[i_a = i_c, i_b = i_d] - P[(abcd)]$. Further, note that when conditioned on specific values of $N_i$'s and $M_i$'s, $i_a=i_c$ and $i_b = i_d$ events become independent. That is, for fixed $(N_i)$ and $(M_i)$ values we have
	\begin{gather*}
	\begin{split}
	P[i_a = i_b, i_c = i_d \mid (N_i), (M_i)] &=
	\frac{\sum_{j}(N_j^2)}{(\sum_{j}{N_j})^2} \frac{\sum_{j}(M_j^2)}{(\sum_{j}{M_j})^2} 
	\end{split}\\
	P[i_a = i_c, i_b = i_d \mid (N_i), (M_i)] =
	\frac{\sum_{j}(N_j M_j)}{(\sum_{j}{N_j})(\sum_{j}{M_j})} \frac{\sum_{j}(N_j M_j)}{(\sum_{j}{N_j})(\sum_{j}{M_j})}. 
	\end{gather*}
	Then by Cauchy-Schwartz, $(\sum_{j}(N_j M_j))^2 \le \sum_{j}(N_j^2)\sum_{j}(M_j^2)$ and therefore $P[i_a = i_b, i_c = i_d \mid (N_i), (M_i)] \ge P[i_a = i_c, i_b = i_d \mid (N_i), (M_i)]$ for any realization of $(N_i)$ and $(M_i)$ variables. That is, $P[(ab,cd)] \ge P[(ac,bd)]$.
\end{proof}
Using the above results, we have
\[
P[\qab \in G, I] = P[(ab,cd)] \ge P[(ac,bd)] = P[\qac \in G, I].
\]

\newcommand{\cAB}{AB}
\newcommand{\cAC}{AC}

\subsubsection{Case $\boldsymbol{I = (ab,c,d)\lor(cd,a,b)\lor(ac,b,d)\lor(bd,a,c)}$} \label{sec:(ab,c,d)}
We prove that
\[
\begin{split}
P[\qab \in G, I]&=P[\qab \in G \mid (ab,c,d) \lor (cd,a,b)]P[(ab,c,d) \lor (cd,a,b)]\\
&\ \ + P[\qab \in G \mid (ac,b,d) \lor (bd,a,c)]P[(ac,b,d) \lor (bd,a,c)]\\ 
&\ge P[\qac \in G \mid (ab,c,d) \lor (cd,a,b)]P[(ab,c,d) \lor (cd,a,b)]\\
&\ \ +P[\qac \in G \mid (ac,b,d) \lor (bd,a,c)]P[(ac,b,d) \lor (bd,a,c)] = P[\qac \in G, I].
\end{split}
\]
From now on, for convenience, we will denote the event $(ab,c,d)\lor(cd,a,b)$ by $\boldsymbol{\cAB}$ and the event $(ac,b,d)\lor(bd,a,c)$ by $\boldsymbol{\cAC}$.

Consider the following results.

\todo{AM: Note that actually a strict inequality holds in this lemma, but not necessary for the proof.}
\begin{lemma}\label{lem:3lin}
	$P[\qab \in G \mid \cAB)] \ge P[\qac \in G \mid \cAC]$.
\end{lemma}
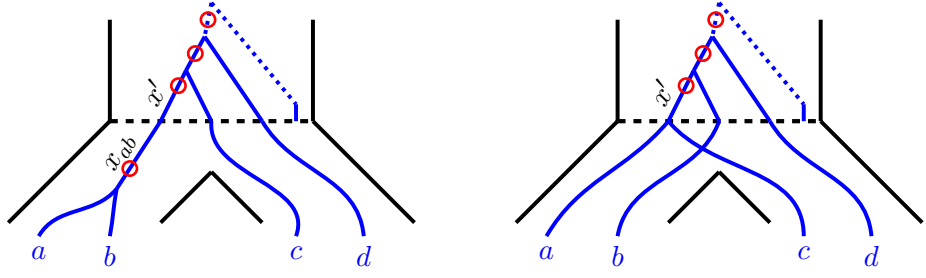
\begin{figure}
	\centering
	\begin{tikzpicture}[scale=0.45]
	\tikzstyle{vertex} = [circle,draw,fill,inner sep=0pt, minimum size=4pt]
	\tikzstyle{edge} = [draw,line width=1.5pt,-]
	\tikzstyle{ledge} = [draw,line width=1.5pt,-,color=blue]
	\tikzstyle{hedge} = [draw,ultra thick,-,color=blue]
	\tikzset{vlabel/.style 2 args={#1=1pt of #2}}
	\def\offset{0}
	
	\foreach \x/\y/\name in {0/0/llr,0/3/ltr,6/0/rlr,6/3/rtr,3/3.5/lr,-3/-3/lab,1.5/-3/rab,4.5/-3/lcd,9/-3/rcd,3/-1.5/rs,0.2/-2/cab,1.5/0/brl1,3/0/brl2,4.5/0/brl3,2.8/2.5/irl1,2.25/1.5/irl2} {
		\coordinate (\name) at (\offset+\x,\y) {};
	}
	
	\foreach \u/\v in {ltr/llr,rtr/rlr,llr/lab,rlr/rcd,rs/rab,rs/lcd} {
		\path [edge] (\u) -- (\v);
	}
	\draw [ultra thick,dashed] (llr) -- (rlr);
	\path [ledge] (brl1) -- node[above=-3pt,color=black,rotate=60] {$x'$} (irl2) node[pos=0.7] (dab') {};
	\path [ledge] (brl2) -- (irl2);
	\path [ledge] (brl3) -- (irl1);
	\path [ledge] (irl2) -- (irl1) node[pos=0.5] (dabc) {};
	\path [ledge,dotted] (lr) -- (irl1) node[pos=0.5] (dr) {};
	\foreach \x/\name in {5.5/rl4} {
		\coordinate (b\name) at (\offset+\x,0) {};
		\coordinate (t\name) at (\offset+\x,0.5) {};
		\path [ledge] (b\name) -- (t\name); 
		\path [ledge,dotted] (lr) -- (t\name);
	}
	\foreach \x/\name/\label in {-2.1/a/$a$,0/b/$b$,5.5/c/$c$,7.5/d/$d$} {
		\coordinate [] (\name) at (\offset+\x,-3.4) {};
		\node [color=blue, below=0pt of \name] {\label};
	}
	\path [ledge] (a) to[out=60,in=-130] (cab);
	\path [ledge] (b) to[out=80,in=-100] (cab);
	\path [ledge] (c) to[out=70,in=-95] (brl2);
	\path [ledge] (d) to[out=95,in=-60] (brl3);
	\path [ledge] (cab) -- node[above=-3pt,rotate=60,color=black] {$x_{ab}$} (brl1) node[pos=0.3] (dab) {};
	
	\foreach \loc in {dab, dab', dabc, dr} {
		\draw [line width=1pt, line cap=round,color=red] (\loc) circle (6pt);
	}
	
	%--------------
	\def\offset{15}
	
	\foreach \x/\y/\name in {0/0/llr,0/3/ltr,6/0/rlr,6/3/rtr,3/3.5/lr,-3/-3/lab,1.5/-3/rab,4.5/-3/lcd,9/-3/rcd,3/-1.5/rs,0.2/-2/cab,1.5/0/brl1,3/0/brl2,4.5/0/brl3,2.8/2.5/irl1,2.25/1.5/irl2} {
		\coordinate (\name) at (\offset+\x,\y) {};
	}
	
	\foreach \u/\v in {ltr/llr,rtr/rlr,llr/lab,rlr/rcd,rs/rab,rs/lcd} {
		\path [edge] (\u) -- (\v);
	}
	\draw [ultra thick,dashed] (llr) -- (rlr);
	\path [ledge] (brl1) -- node[above=-3pt,color=black,rotate=60] {$x'$} (irl2) node[pos=0.7] (dab') {};
	\path [ledge] (brl2) -- (irl2);
	\path [ledge] (brl3) -- (irl1);
	\path [ledge] (irl2) -- (irl1) node[pos=0.5] (dabc) {};
	\path [ledge,dotted] (lr) -- (irl1) node[pos=0.5] (dr) {};
	\foreach \x/\name in {5.5/rl4} {
		\coordinate (b\name) at (\offset+\x,0) {};
		\coordinate (t\name) at (\offset+\x,0.5) {};
		\path [ledge] (b\name) -- (t\name); 
		\path [ledge,dotted] (lr) -- (t\name);
	}
	\foreach \x/\name/\label in {-2.1/a/$a$,0/b/$b$,5.5/c/$c$,7.5/d/$d$} {
		\coordinate [] (\name) at (\offset+\x,-3.4) {};
		\node [color=blue, below=0pt of \name] {\label};
	}
	\path [ledge] (a) to[out=60,in=-130] (brl1);
	\path [ledge] (b) to[out=80,in=-110] (brl2);
	\path [ledge] (c) to[out=90,in=-50] (brl1);
	\path [ledge] (d) to[out=95,in=-60] (brl3);
	
	\foreach \loc in {dab', dabc, dr} {
		\draw [line width=1pt, line cap=round,color=red] (\loc) circle (6pt);
	}
\end{tikzpicture}
	\caption{Caterpillar locus trees $L_{(ab,c,d)}$ (left) and $L_{(ac,b,d)}$ (right) embedded into the species tree. The red circles represent the potential duplication locations that could influence the gene tree probabilities. Note that the $\cL_r$ scenarios in the root edges are identical. That is, $x'$ lengths are equal and the duplication locations above the dashed speciation lines are identical}
	\label{fig:3lin-1}
\end{figure}
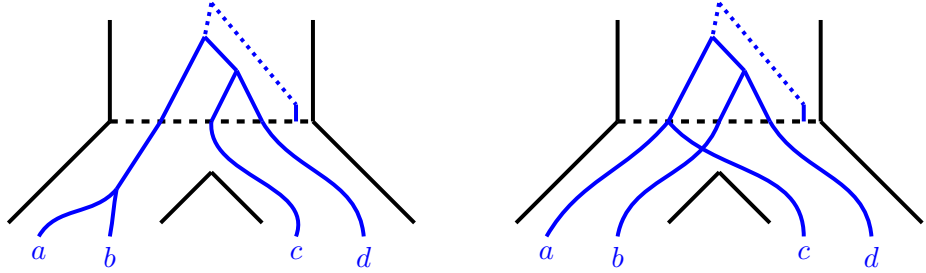
\begin{figure}
	\centering
	\begin{tikzpicture}[scale=0.45]
	\tikzstyle{vertex} = [circle,draw,fill,inner sep=0pt, minimum size=4pt]
	\tikzstyle{edge} = [draw,line width=1.5pt,-]
	\tikzstyle{ledge} = [draw,line width=1.5pt,-,color=blue]
	\tikzstyle{hedge} = [draw,ultra thick,-,color=blue]
	\tikzset{vlabel/.style 2 args={#1=1pt of #2}}
	\def\offset{0}
	
	\foreach \x/\y/\name in {0/0/llr,0/3/ltr,6/0/rlr,6/3/rtr,3/3.5/lr,-3/-3/lab,1.5/-3/rab,4.5/-3/lcd,9/-3/rcd,3/-1.5/rs,0.2/-2/cab,1.5/0/brl1,3/0/brl2,4.5/0/brl3,2.8/2.5/irl1,3.75/1.5/irl2} {
		\coordinate (\name) at (\offset+\x,\y) {};
	}
	
	\foreach \u/\v in {ltr/llr,rtr/rlr,llr/lab,rlr/rcd,rs/rab,rs/lcd} {
		\path [edge] (\u) -- (\v);
	}
	\draw [ultra thick,dashed] (llr) -- (rlr);
	\path [ledge] (brl1) -- (irl1);
	\path [ledge] (brl2) -- (irl2);
	\path [ledge] (brl3) -- (irl2);
	\path [ledge] (irl2) -- (irl1);
	\path [ledge,dotted] (lr) -- (irl1);
	\foreach \x/\name in {5.5/rl4} {
		\coordinate (b\name) at (\offset+\x,0) {};
		\coordinate (t\name) at (\offset+\x,0.5) {};
		\path [ledge] (b\name) -- (t\name); 
		\path [ledge,dotted] (lr) -- (t\name);
	}
	\foreach \x/\name/\label in {-2.1/a/$a$,0/b/$b$,5.5/c/$c$,7.5/d/$d$} {
		\coordinate [] (\name) at (\offset+\x,-3.4) {};
		\node [color=blue, below=0pt of \name] {\label};
	}
	\path [ledge] (a) to[out=60,in=-130] (cab);
	\path [ledge] (b) to[out=80,in=-100] (cab);
	\path [ledge] (c) to[out=70,in=-95] (brl2);
	\path [ledge] (d) to[out=95,in=-60] (brl3);
	\path [ledge] (cab) -- (brl1);
	
	%--------------
	\def\offset{15}
	
	\foreach \x/\y/\name in {0/0/llr,0/3/ltr,6/0/rlr,6/3/rtr,3/3.5/lr,-3/-3/lab,1.5/-3/rab,4.5/-3/lcd,9/-3/rcd,3/-1.5/rs,0.2/-2/cab,1.5/0/brl1,3/0/brl2,4.5/0/brl3,2.8/2.5/irl1,3.75/1.5/irl2} {
		\coordinate (\name) at (\offset+\x,\y) {};
	}
	
	\foreach \u/\v in {ltr/llr,rtr/rlr,llr/lab,rlr/rcd,rs/rab,rs/lcd} {
		\path [edge] (\u) -- (\v);
	}
	\draw [ultra thick,dashed] (llr) -- (rlr);
	\path [ledge] (brl1) -- (irl1);
	\path [ledge] (brl2) -- (irl2);
	\path [ledge] (brl3) -- (irl2);
	\path [ledge] (irl2) -- (irl1);
	\path [ledge,dotted] (lr) -- (irl1);
	\foreach \x/\name in {5.5/rl4} {
		\coordinate (b\name) at (\offset+\x,0) {};
		\coordinate (t\name) at (\offset+\x,0.5) {};
		\path [ledge] (b\name) -- (t\name); 
		\path [ledge,dotted] (lr) -- (t\name);
	}
	\foreach \x/\name/\label in {-2.1/a/$a$,0/b/$b$,5.5/c/$c$,7.5/d/$d$} {
		\coordinate [] (\name) at (\offset+\x,-3.4) {};
		\node [color=blue, below=0pt of \name] {\label};
	}
	\path [ledge] (a) to[out=60,in=-130] (brl1);
	\path [ledge] (b) to[out=80,in=-110] (brl2);
	\path [ledge] (c) to[out=90,in=-50] (brl1);
	\path [ledge] (d) to[out=95,in=-60] (brl3);
\end{tikzpicture}
	\caption{Balanced locus trees $L_{(ab,c,d)}$ (left) and $L_{(ac,b,d)}$ (right) embedded into the species tree.}
	\label{fig:3lin-2}
\end{figure}
\begin{proof}
	Note that fixing the number of root lineages allows us to treat the dup-loss processes independently for the root edge and for the lower edges. Let $\cL_r$ be a dup-loss scenario in the root edge conditioned on $RL = l$. 
	Without loss of generality then assume that in case $(ab,c,d)$, we have $i_a = i_b = 1$, $i_c=2$, and $i_d=3$; in case $(cd,a,b)$ we assume $i_c = i_d = 1$, $i_a = 2$, and $i_b = 3$. Similarly, under $(ac,b,d)$ we assume $i_a = i_c = 1, i_b = 2, i_d = 3$ and under $(bd,a,c)$ we assume that $i_b = i_d = 1, i_a = 2, i_d = 3$. 
	Then a fixed $\cL_r$ scenario similarly influences the locus tree in all four cases.
	
	Given that $(ab,c,d)$ and $(cd,a,b)$ cases are virtually identical for the rest of the proof (since they are symmetric), for simplicity, we will only consider the $(ab,c,d)$ case. Similarly, under the $\cAC$ event, we will only consider case $(ac,b,d)$.
	
	Then, Figures~\ref{fig:3lin-1} and \ref{fig:3lin-2} depict two possible topologies of the $\cL_r$ scenario when acting on the root lineages $1,2,$ and $3$. Observe that the third topology, where root lineages 1 and 3 form a cherry, is identical in terms of analysis to the topology depicted in Figure~\ref{fig:3lin-1}, and therefore is not considered.
	
	Note now that in Figure~\ref{fig:3lin-1} the resulting locus trees $L_{(ab,c,d)}$ and $L_{(ac,b,d)}$ are both caterpillars, while in Figure~\ref{fig:3lin-2} the locus trees are both balanced. We now consider these two cases individually.
	
	\begin{itemize}
		\item[(i)] \textbf{$L_{(ab,c,d)}$ and $L_{(ac,b,d)}$ are caterpillars} (see Figure~\ref{fig:3lin-1}). 
		Let $x_{ab}$ be the distance (in coalescent units) from the root speciation event to the coalescence of $a$ and $b$ under the $(ab,c,d)$ case (as shown on the figure)). Note that $x_{ab} \ge 0$.
		The there are two cases to consider.
		\begin{itemize}
			\item \textit{There is a duplication along the $x_{ab}$ lineage.} Then, as shown in section~\ref{sec:cat-Xdup}, $P[\qab \in G \mid AB, \cL_r] = 1$. That is, $P[\qab \in G \mid AB, \cL_r] \ge P[\qac \in G \mid AC, \cL_r]$.
			
			\item \textit{No duplications along the $x_{ab}$ lineage.}
			Since $L_{(ab,c,d)}$ and $L_{(ac,b,d)}$ are both caterpillars, we denote their edges by $X$ and $Y$ as shown in Figure~\ref{fig:q-cat}. Then $w(X_{(ab,c,d)}) = x' + x_{ab}$, whereas $w(X_{(ac,b,d)}) = x'$. Further, the two locus trees are identical in terms of the duplication locations in their internal edges.
			
			Then, by Observation~\ref{obs:grow}, it is not difficult to see that $P[\qab \in G \mid \cL_r, (ab,cd)] \ge P[\qac \in G \mid \cL_r, (ac,bd)]$ for any fixed $\cL_r$. Therefore, the lemma holds.
		\end{itemize}
		
		\item[(ii)] \textbf{$L_{(ab,c,d)}$ and $L_{(ac,b,d)}$ are balanced} (see Figure~\ref{fig:3lin-2}). By Lemma~\ref{lem:balanced-lt}, $P[\qab \in G \mid AB, \cL_r] = 1$ and $P[\qac \in G \mid AC, \cL_r] = 1$. That is, the lemma holds.
	\end{itemize} 
\end{proof}

\begin{lemma}\label{lem:3lin-2}
	$P[\qab \in G \mid \cAC] \ge P[\qac \in G \mid \cAB]$.
\end{lemma}
\begin{proof}
	This result follows from Lemma~\ref{lem:3lin} and the following relations:
	\begin{gather*}
	2P[\qab \in G \mid \cAC] + P[\qac \in G \mid \cAC] = 1\\
	2P[\qac \in G \mid \cAB] + P[\qab \in G \mid \cAB] = 1.
	\end{gather*}
\end{proof}
\begin{observ}\label{obs:3lin-3}
	By Lemma~\ref{lem:greater}, $P[\qac \in G \mid \cAC] \ge P[\qab \in G \mid \cAC]$. Then combining this with Lemma~\ref{lem:3lin-2} we have $P[\qac \in G \mid \cAC] \ge P[\qac \in G \mid \cAB]$
\end{observ}

\begin{lemma}\label{lem:(ab,c,d)}
	$P[\cAB] \ge P[\cAC]$.
\end{lemma}
\begin{proof}
	Observe that
	\begin{gather*}
	\begin{split}
	P[\cAB] = P[(ab,c,d) \lor (cd,a,b)] &= P[i_a = i_b, i_c \ne i_d] - P[(abc,d)] - P[(abd,c)]\\
	&+ P[i_c = i_d, i_a \ne i_b] - P[(acd,b)] - P[(bcd,a)];
	\end{split}\\
	\begin{split}
	P[\cAC] = P[(ac,b,d) \lor (bd,a,c)] &= P[i_a = i_c, i_b \ne i_d] - P[(abc,d)] - P[(acd,b)]\\
	&+ P[i_b = i_d, i_a \ne i_c] - P[(abd,c)] - P[(bcd,a)].
	\end{split}
	\end{gather*}
	Therefore, it is sufficient to show that
	\begin{equation*}
	P[i_a = i_b, i_c \ne i_d] + P[i_c = i_d, i_a \ne i_b] \ge P[i_a = i_c, i_b \ne i_d] + P[i_b = i_d, i_a \ne i_c].
	\end{equation*}
	
	Let $x := P[i_a = i_b]$ and $y := P[i_c = i_d]]$. Note that Legried et al. demonstrated that $x, y \ge \frac{1}{l}$~\cite{Legried:2019dl}. Then
	\begin{equation}\label{eqn:ab,c,d}
	P[i_a = i_b, i_c \ne i_d] + P[i_c = i_d, i_a \ne i_b] = x(1-y) + y(1-x) = x + y -2xy.
	%&\ge x + y - x^2 - y^2 = x(1-x) + y(1-y).
	\end{equation}
	Further,
	\[
	\begin{split}
	&P[i_b = i_d \mid i_a = i_c] = \sum_{j=1}^{l}P[i_b = i_d \mid i_a = i_c = j]P[i_a = i_c = j \mid i_a = i_c]\\
	&= \frac{1}{l}\sum_{j=1}^{l}\sum_{k=1}^{l}P[i_b = i_d = k \mid i_a = i_c = j] = \frac{1}{l}\sum_{j=1}^{l}\sum_{k=1}^{l}P[i_b = k \mid i_a = j]P[i_d = k \mid i_c = j]\\
	&= \frac{1}{l}l\Big(P[i_b = 1 \mid i_a = 1]P[i_d = 1 \mid i_c = 1] + \ldots + P[i_b = l \mid i_a = 1]P[i_d = l \mid i_c = 1]\Big)\\
	&= xy + (l-1)\frac{(1-x)}{(l-1)}\frac{(1-y)}{(l-1)}.
	\end{split}
	\]
	Last equality is due to $P[i_b = 1\mid i_a = 1] = x$ and $P[i_d = 1 \mid i_c = 1] = y$.
	Then 
	\[
	P[i_a=i_c, i_b \ne i_d] = (1- P[i_b = i_d \mid i_a = i_c])P[i_a = i_c] = (1 - xy - \frac{(1-x)(1-y)}{(l-1)})\frac{1}{l}
	\]
	\begin{equation}\label{eqn:ac,b,d}
	P[i_a=i_c, i_b \ne i_d] + P[i_b=i_d, i_a \ne i_c] = \frac{2}{l(l-1)}(l - 2 - lxy + x + y).
	\end{equation}
	
	Multiplying equations~\ref{eqn:ab,c,d} and \ref{eqn:ac,b,d} by $l(l-1)$ and fixing some $y \in [1/l, 1]$ we get two \textbf{linear} functions.
	\begin{gather*}
	f(x) := l(l-1)(x + y - 2xy)\\
	g(x) := 2(l - 2 - lxy + x + y).
	\end{gather*}
	It is then sufficient to show that $f(1/l) \ge g(1/l)$ and $f(1) \ge g(1)$ to conclude the proof (since $x$ is in the $[1/l, 1]$ range).
	\begin{gather*}
	f(1/l) = l - 1 + y(l-1)(l-2);\\
	g(1/l) = 2l - 4 - 2y + 2/l + 2y = 2l - 4+ 2/l.
	\end{gather*}
	Observe that $f(1/l)$ is minimum when $y = 1/l$ (since that is the smallest possible value for $y$). In that case $f(1/l) = l - 1 + (l^2 - 3l + 2)/l = 2l - 4 + 2/l$. That is, $f(1/l) \ge g(1/l)$ for all values of $y$. Let us now compare $f(1)$ and $g(1)$.
	\begin{gather*}
	f(1) = l(l-1)(1 - y);\\
	g(1) = 2 (l - 2 - ly + 1 + y) = 2(l - 1 - y(l-1)) = 2(l-1)(1-y).
	\end{gather*}
	Note that, e.g., for the $(ab,c,d)$ case to be feasible, we need to have $l \ge 3$. Therefore, $f(1) \ge g(1)$.	
\end{proof}

Summarizing the above results we have.
\[
\begin{split}
&P[\qab \in G \mid \cAB]P[\cAB] + P[\qab \in G \mid \cAC]P[\cAC]\\
&\ge P[\qac \in G \mid \cAC]P[\cAB] + P[\qac \in G \mid \cAB]P[\cAC]\\
&\ge P[\qac \in G \mid \cAC]P[\cAC] + P[\qac \in G \mid \cAB]P[\cAB].
\end{split}
\]
Note that the first inequality is due to Lemmas~\ref{lem:3lin} and \ref{lem:3lin-2}.
The last inequality is due to Lemma~\ref{lem:(ab,c,d)} and Observation~\ref{obs:3lin-3}.

That is, our main statement holds.

\subsubsection{Case $\boldsymbol{I = (abc,d)\lor(abd,c)\lor(acd,b)\lor(bdc,a)\lor(abcd)}$.}
In all four cases locus tree $L$ displays the quartet $\qab$. Therefore, by Lemma~\ref{lem:greater} $P[\qab \in G \mid I] > P[\qac \in G \mid I]$. Observe that we obtain the strict inequality in this case.

\subsubsection{Case $\boldsymbol{I = (ad,bc)\lor(ad,b,c)\lor(bc,a,d)}$.}
It is not difficult to see that in this case, $L$ displays quartet $\qad$. Therefore (as can be seen from the derivations in Section~\ref{sec:probs}), $P[\qab \in G \mid I] = P[\qac \in G \mid I]$.

This concludes the proof for balanced $S$.

\subsection{$S$ is a caterpillar}
Without lost of generality assume that $S$ is as appears in Figure~\ref{fig:embed-cat}. Similarly to the balanced case, we implicitly condition the probability space on a fixed number of loci (lineages) existing at the moment of speciation as shown in the figure. Note that, while in the balanced case we considered root lineages, in the caterpillar scenario we consider lineages at the least common ancestor of $A,B,$ and $C$. That is, we refer to these lineages/loci as ABC-lineages. Finally, as in the balanced case, we denote the number of ABC-lineages by $l$.

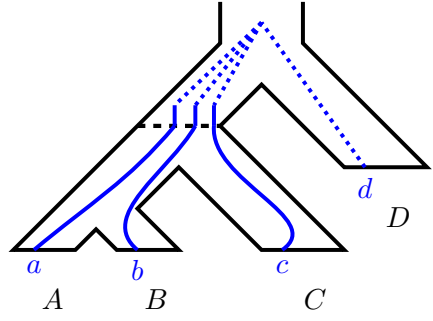
\begin{figure}
	\centering
	\begin{tikzpicture}[scale=0.55]
	\tikzstyle{vertex} = [circle,draw,fill,inner sep=0pt, minimum size=4pt]
	\tikzstyle{edge} = [draw,line width=1.5pt,-]
	\tikzstyle{ledge} = [draw,line width=1.5pt,-,color=blue]
	\tikzstyle{hedge} = [draw,ultra thick,-,color=blue]
	\tikzset{vlabel/.style 2 args={#1=1pt of #2}}
	\def\offset{0}
	
	\foreach \x/\y/\name in {0/0/labc,2/0/rabc,2/2/lbr,2/3/ltr,4/2/rbr,4/3/rtr,3/1/sr,1/-1/sabc,-1/-2.5/sab,3/2.5/rl,-3/-3/la,-1.5/-3/ra,-0.5/-3/lb,1/-3/rb,3/-3/lc,5/-3/rc,5/-1/ld,7/-1/rd,0/-2/rsab} {
		\coordinate (\name) at (\offset+\x,\y) {};
	}
	
	\path [edge] (ltr) -- (lbr) -- (la) -- (ra) -- (sab) -- (lb) -- (rb) -- (rsab) -- (sabc) -- (lc) -- (rc) -- (rabc) -- (sr) -- (ld) -- (rd) -- (rbr) -- (rtr);
	
	\draw [ultra thick,dashed] (labc) -- (rabc);
	\foreach \x/\y/\slabel in {-2.35/-3/$A$,0.15/-3/$B$,4/-3/$C$,6/-1/$D$} {
		\node[] at (\offset+\x+0.3,\y - 1.2) {\slabel};
	}

	\foreach \x/\name in {0.9/rl1,1.4/rl2,1.85/rl3} {
		\coordinate (b\name) at (\offset+\x,0) {};
		\coordinate (t\name) at (\offset+\x,0.5) {};
		\path [ledge] (b\name) -- (t\name); 
		\path [ledge,dotted] (rl) -- (t\name);
	}
	\foreach \x/\y/\name/\label in {-2.5/-3/a/$a$,0/-3/b/$b$,3.5/-3/c/$c$,5.5/-1/d/$d$} {
		\coordinate [] (\name) at (\offset+\x,\y) {};
		\node [color=blue, below=0pt of \name] {\label};
	}
	\path [ledge] (a) to[out=40,in=-130] (brl1);
	\path [ledge] (b) to[out=140,in=-120] (brl2);
	\path [ledge] (c) to[out=35,in=-95] (brl3);
	\path [ledge,dotted] (d) -- (rl);
\end{tikzpicture}
	\caption{An example of the locus tree embedding into a caterpillar species tree. The three locus lineages crossing the dashed speciation line are the ABC-lineages}
	\label{fig:embed-cat}
\end{figure}

We then use the $i_a, i_b, i_c$ notation in the same way as in the previous section (while referring to indices of ABC-lineages). Further, $\cI = \set{(a,b,c); (ab,c); (ac,b); (bc,a); (abc)}$ scenarios describe relations between $i_a, i_b,$ and $i_c$.

We now prove that $P[\qab \in G, I] \ge P[\qac \in G, I]$ for all $I$ in $\set{(a,b,c);(ab,c)\lor(ac,b);(bc,a);(abc)}$. Moreover, for at least one such $I$, the strict inequality holds; in particular, see case~\ref{sec:(abc)} below.

\subsubsection{Case $\boldsymbol{I = (a,b,c)}$.}\label{sec:(a,b,c)}
By symmetry of the DLCoal model, we have $P[\qab \in G \mid I] = P[\qac \in G \mid I] = P[\qad \in G \mid I]$. That is, reshuffling the $i_a, i_b, i_c$ labels will not affect the probabilities.

Then, $P[\qab \in G, I] = P[\qab \in G \mid I]P[I] = P[\qac \in G \mid I]P[I] = P[\qac \in G, I]$.

\subsubsection{Case $\boldsymbol{I = (ab,c) \lor (ac,b)}$.}
The proof in this case is similar to case~\ref{sec:(ab,c,d)} for balanced $S$. In particular, observe the following.

\begin{lemma}\label{lem:(ab,c)-1}
	$P[(ab,c)] \ge P[(ac,b)]$.	
\end{lemma}
\begin{proof}
	It is sufficient to show that $P[i_a = i_b] \ge P[i_a = i_c]$. Note that $i_a$ and $i_c$ are independent and therefore $P[i_a = i_c] = 1/l$. Further, Legried et al.~\cite{Legried:2019dl} showed that $P[i_a = i_b] \ge 1/l$.
\end{proof}

\begin{lemma}\label{lem:(ab,c)-2} $ $\\\vspace{-1em}
	\begin{itemize}
		\item[(i)] $P[\qab \in G \mid (ab,c)] \ge P[\qac \in G \mid (ac,b)]$;
		\item[(ii)] $P[\qab \in G \mid (ac,b)] \ge P[\qac \in G \mid (ab,c)]$;
		\item[(iii)] $P[\qac \in G \mid (ac,b)] \ge P[\qac \in G \mid (ab,c)]$.
	\end{itemize}
\end{lemma}
\begin{proof}
	Note that (i) corresponds to Lemma~\ref{lem:3lin}, (ii) corresponds to Lemma~\ref{lem:3lin-2}, and (iii) corresponds to Observation~\ref{obs:3lin-3} from Section~\ref{sec:(ab,c,d)}.
	The proofs of these statements are very similar, so we omit them for brevity.
\end{proof}
Then, similarly to Section~\ref{sec:(ab,c,d)} we have
\[
\begin{split}
P[\qab \in G, I] &= P[\qab \in G \mid (ab,c)]P[(ab,c)] + P[\qab \in G \mid (ac,b)]P[(ac,b)]\\
&\ge P[\qac \in G \mid (ac,b)]P[(ab,c)] + P[\qac \in G \mid (ab,c)]P[(ac,b)]\\
&\ge P[\qac \in G \mid (ac,b)]P[(ac,b)] + P[\qac \in G \mid (ab,c)]P[(ab,c)] = P[\qac \in G, I].
\end{split}
\]
%\begin{gather*}
%	P[\qab \in G, I] = \frac{P[\qab \in G \mid (ab,c)]P[(ab,c)] + P[\qab \in G \mid (ac,b)]P[(ac,b)]}{P[I]} = \frac{P[(ab,c)]}{P[I]}\\
%	P[\qac \in G \mid I] = \frac{P[\qac \in G \mid (ab,c)]P[(ab,c)] + P[\qac \in G \mid (ac,b)]P[(ac,b)]}{P[I]} = \frac{P[(ac,b)]}{P[I]}
%\end{gather*}
%That is, by Lemma~\ref{lem:(ab,c)-1}, $P[\qab \in G \mid I] \ge P[\qac \in G \mid I]$.

\subsubsection{Case $\boldsymbol{I = (bc,a)}$.}
In this case $P[\qab \in G \mid I] = P[\qac \in G \mid I]$, since the locus tree displays the third quartet, $\qad$.

\subsubsection{Case $\boldsymbol{I = (abc)}$.}\label{sec:(abc)}
The locus tree displays quartet $\qab$; therefore, by Lemma~\ref{lem:greater} and the law of total probability, $P[\qab \in G \mid I] > P[\qac \in G \mid I]$.

\section{Acknowledgments} This material is based upon work supported by the National Science Foundation under Grant No. 1617626.

\bibliographystyle{abbrv}
\bibliography{complete-no-doi}

\end{document}